\documentclass[12pt]{article}%
\usepackage{graphicx}
\usepackage{amsmath}
\usepackage{amsfonts}
\usepackage{amssymb}%
\providecommand{\U}[1]{\protect\rule{.1in}{.1in}}
\newtheorem{theorem}{Theorem}

\newtheorem{lemma}[theorem]{Lemma}

\newtheorem{proposition}[theorem]{Proposition}

\newenvironment{proof}[1][Proof]{\noindent\textbf{#1.} }{\ \rule{0.5em}{0.5em}}
\begin{document}

\title{\textbf{A continuum of pure states in the Ising model on a halfplane
\footnotetext{Part of this work has been carried out in the framework of the
Labex Archimede (ANR-11-LABX-0033) and of the A*MIDEX project
(ANR-11-IDEX-0001-02), funded by the \textquotedblleft Investissements
d'Avenir" French Government program managed by the French National Research
Agency (ANR). Part of this work has been carried out by SS at IITP RAS. The
support of Russian Foundation for Sciences (project No. 14-50-00150) is
gratefully acknowledged. The research of CMN was supported in part by US-NSF
grant DMS-1507019. The authors thank Alessio Squarcini for help with TeX
issues; {they also thank an anonymous referee for useful comments and
suggestions.} }}}
\author{Douglas Abraham$^{*}$, Charles M. Newman$^{\dag,\dag\dag}$, Senya
Shlosman$^{\natural,\sharp,\flat} $\\$^{*}$Rudolf Peierls Center for Theoretical Physics, \\University of Oxford, Oxford \\$^{\dag}$Courant Institute of Mathematical Sciences \\of New York University, New York; \\$^{\dag\dag}$NYU-ECNU Institute of Mathematical Sciences \\at NYU-Shanghai, Shanghai \\$^{\natural}$Skolkovo Institute of Science and Technology, Moscow; \\$^{\sharp}$Aix Marseille Univ, Universit\'{e} de Toulon, \\CNRS, CPT, Marseille; \\$^{\flat}$Inst. of the Information Transmission Problems,\\RAS, Moscow}
\maketitle

\begin{center}
\textit{We join with the other contributors to this special issue of JSP to}

\textit{acknowledge the many contributions to Mathematical }

\textit{and Statistical Physics made over the years by }

\textit{\textbf{Juerg Froehlich}, \textbf{Tom Spencer} and \textbf{Herbert
Spohn}.}

\textit{May they each live and be well to one hundred twenty.\pagebreak}
\end{center}

\begin{abstract}
We study the homogeneous nearest-neighbor Ising ferromagnet on the right half
plane with a Dobrushin type boundary condition --- say plus on the top part of
the boundary and minus on the bottom. For sufficiently low temperature $T$, we
completely characterize the pure (i.e., extremal) Gibbs states, as follows.
There is exactly one for each angle $\theta\in\lbrack-\pi/2,+\pi/2]$; here
$\theta$ specifies the asymptotic angle of the interface separating regions
where the spin configuration looks like that of the plus (respectively, minus)
full-plane state. Some of these conclusions are extended all the way to
$T=T_{c}$ by developing new Ising exact solution results -- in particular,
there is at least one pure state for each $\theta$.

\end{abstract}

\section{Introduction}

\label{introduction}

In equilibrium statistical mechanics and particularly in the study of phase
transitions, an important role is played by the collection of Gibbs states (or
measures or distributions) at temperature $T$ of a fully infinite system, such
as the standard (homogeneous nearest-neighbor) Ising ferromagnet on an
infinite lattice, such as ${\mathbb{Z}}^{d}$. The precise mathematical
formulation of such Gibbs states was pioneered by Dobrushin, Lanford and
Ruelle \cite{D1, LR}.

A nice discussion of the connection between the phenomenon of phase
transitions and nonuniqueness of infinite-volume Gibbs distributions may be
found in the Introduction section of \cite{G}. The pure phases of a physical
system correspond to the extremal (also called pure) Gibbs states --- the ones
which cannot be decomposed further as convex combinations of other states.
Thus for models such as the standard Ising ferromagnet on ${\mathbb{Z}}^{d}$,
considerable effort has been devoted to characterizing the collection of all
pure states.

One of the major accomplishments in the field is a complete characterization
for ${\mathbb{Z}}^{2}$ \cite{A, H} that there are exactly two pure states for
$T<T_{c}$ (corresponding respectively to plus and to minus boundary
conditions). Another major accomplishment was the demonstration by Dobrushin
\cite{D2} that for ${\mathbb{Z}}^{d}$ with $d\geq3$ at low enough $T$, there
are other, non-translation invariant, states which display an interface
(parallel to a coordinate axis). Together with two translation-invariant
states they form a countable set of extremal states. Other results concern the
analysis of pure states on infinite graphs such as homogeneous trees \cite{B,
BG, GRS, GMRS}.

The subject of this paper is one which has not been much investigated
previously --- namely, the analysis of the pure states when the underlying
lattice or graph is infinite, but with a nontrivial boundary. In particular,
we consider a half plane --- say, the right half-plane. Here one may specify a
boundary condition on the (left) boundary of the half-plane. In the absence of
boundary conditions (i.e., with free boundary conditions), the pure states
would mimic those of the full plane --- a single pure state for $T\geq T_{c}$
and exactly two for $T<T_{c}$.

If one instead specified a plus (or similarly, minus) left boundary, there
would be a unique pure state for all $T$. Instead we use a Dobrushin type
boundary condition --- plus on the top and minus on the bottom of the left
boundary. Our results prove that in this case, there are uncountably many pure
states -- exactly one for for each $\theta\in\lbrack-\pi/2,+\pi/2]$ where
$\theta$ is the asymptotic angle of an interface that starts at the origin
between plus and minus like regions of the half-plane. We think these are all
the pure states of our model, though we prove this only for temperatures low
enough. Similar results were formulated earlier by one of us in \cite{DS} and
called the `Meniscus theorem', but without a proof.


We believe that the situation in higher dimensions~$d$ is somewhat similar but
with possible new phenomena such as needing at least $d-2>1$ continuous
variables to parameterize the half-space pure states. For example in the $d=3$
half-space, $\{(s,t,u):s>0\}$ one may take a boundary condition in the
$(t,u))$-plane consisting of four alternating plus and minus regions --- say,
for simplicity, the four quadrants separated by the two interface lines along
the $t$ and $u$ axes. Then there could be a family of mutually singular Gibbs
states, each with two approximately planar interfaces, emanating from the two
boundary interface lines. It should take two continuous angular variables, say
$\theta_{1}$ and $\theta_{2}$, to parameterize this family, each giving the
asympotic angular deviation of one of the two interfaces from the
corresponding coordinate plane (the $(t,s)$ or the $(u,s)$ plane).

Other models with uncountably many pure states can be constructed on
homogeneous trees that have as much surface as bulk volume (in the
thermodynamic limit), see, e.g., \cite{GMRS}. Examples of a different nature
are provided by models with quasiperiodic order, as considered, e.g., in
\cite{vEM,vEMZ}, and stacked models \cite{WF}. We also note that models in
half-spaces were studied earlier --- see, e.g., \cite{Ba, S}.

The remaining sections of the paper are organized as follows. In
Section~\ref{definition}, we give a precise definition of our halfplane Ising
model and in Section~\ref{theorems} the main results are stated. Proofs for
the simplest case of ground states where $T=0$ are given in
Section~\ref{zerotemp} while the proofs for temperature $T>0$ but small enough
for cluster expansions to apply are presented in Section~\ref{lowtemp}.
Finally, the analysis valid for all $T<T_{c}$ is given in
Sections~\ref{theorems2} and~\ref{belowcritical} of the paper; there we show
that there is at least one pure state with an interface at angle $\theta$ for
each $\theta\in\lbrack-\pi/2,+\pi/2].$ We have no proof that these are all the
pure states.

\section{Definition of the model}

\label{definition}

We consider the Ising model, defined by the usual Ising Hamiltonian
\begin{equation}
H\left(  \sigma\right)  =-\sum_{x,y\text{ n.n. }}\sigma_{x}\sigma
_{y},\label{01}%
\end{equation}
but on $\mathbb{Z}_{+}^{2\ast}=\left\{  x=\left(  s,t\right)  :s,t\in
\mathbb{Z}^{1}+\frac{1}{2},s>0\right\}  $, the halfplane (of the dual
lattice). We will be interested in a specific boundary condition -- which in
our case is a configuration $\sigma^{\pm}$ on $\partial\mathbb{Z}_{+}%
^{2}\equiv\left\{  \left(  -\frac{1}{2},t\right)  :t\in\mathbb{Z}^{1}+\frac
{1}{2}\right\}  $ -- defined by
\[
\sigma_{\left(  -\frac{1}{2},t\right)  }^{\pm}=\left\{
\begin{array}
[c]{cc}%
+1 & \text{if }t>0\\
-1 & \text{if }t<0
\end{array}
\right.  .
\]
The corresponding relative Hamiltonian $H\left(  \sigma{\Huge |}\sigma^{\pm
}\right)  $ is given, accordingly, by%
\[
H\left(  \sigma{\Huge |}\sigma^{\pm}\right)  =-\sum_{\substack{x,y\in
\mathbb{Z}_{+}^{2}\text{ n.n.}\\\text{ }}}\sigma_{x}\sigma_{y}-\sum
_{t>0}\sigma_{\left(  -\frac{1}{2},t\right)  }+\sum_{t<0}\sigma_{\left(
-\frac{1}{2},t\right)  }.
\]
We will also consider the strips%
\[
\mathbb{Z}_{N}^{2}=\left\{  x=\left(  s,t\right)  :s,t\in\mathbb{Z}^{1}%
+\frac{1}{2},0<s<N\right\}  ,
\]
and we will need different boundary conditions on their right boundaries. It
is convenient to view the boundary condition as a configuration on all the
lattice $\mathbb{Z}_{+}^{2},$ and we will use the following family
$\sigma^{\theta}$:%
\[
\sigma_{\left(  s,t\right)  }^{\theta}=\left\{
\begin{array}
[c]{cc}%
+1 & \text{if }\frac{t}{s}\geq\tan\theta\\
-1 & \text{if }\frac{t}{s}<\tan\theta
\end{array}
\right.  ,~\theta\in\left[  -\frac{\pi}{2},\frac{\pi}{2}\right]  .
\]
We denote by $\left\langle \ast\right\rangle _{N}^{\theta}$ the Gibbs state in
the strip $\mathbb{Z}_{N}^{2}$ with the boundary condition $\sigma^{\pm}$ on
its left edge and with the boundary condition $\sigma^{\theta}$ on its right edge.

\section{Main theorems for low temperature}

\label{theorems}

\begin{theorem}
Let the temperature $T=\beta^{-1}$ be low enough. Then for every $\theta
\in\left(  -\frac{\pi}{2},\frac{\pi}{2}\right)  $ there exists a Gibbs state
$\left\langle \ast\right\rangle ^{\theta}$ of the model $\left(
\ref{01}\right)  $ at the temperature $T,$ such that
\begin{equation}
\left\langle \sigma_{\left(  s,t\right)  }\right\rangle ^{\theta}%
\rightarrow\left\{
\begin{array}
[c]{c}%
+m^{\ast}\left(  \beta\right)  \\
-m^{\ast}\left(  \beta\right)
\end{array}
\right\}  \text{ when }s,t\rightarrow\infty\text{ s.t. }\left\{
\begin{array}
[c]{c}%
\lim\inf\\
\lim\sup
\end{array}
\right\}  \left(  \frac{t}{s}\right)  \left\{
\begin{array}
[c]{c}%
>\\
<
\end{array}
\right\}  \tan\theta.\label{1.5}%
\end{equation}
It follows that the states $\left\langle \ast\right\rangle ^{\theta}$ with
different $\theta$'s are mutually singular.
\end{theorem}

The cases $\theta=\pm\frac{\pi}{2}$ require obvious modifications, as given in
the next theorem.

\begin{theorem}
\label{t2} Let the temperature $T=\beta^{-1}$ be low enough. Then for
$\theta=\pm\frac{\pi}{2}$ there exist two Gibbs states $\left\langle
\ast\right\rangle ^{\pm\frac{\pi}{2}}$ of the model $\left(  \ref{01}\right)
$ at temperature $T,$ such that for any $C\in\left(  -\infty,+\infty\right)
$
\begin{equation}
\left\langle \sigma_{\left(  s,t\right)  }\right\rangle ^{\pm\frac{\pi}{2}%
}\rightarrow\mp m^{\ast}\left(  \beta\right)  \text{ for }s,t\rightarrow
\infty\text{ s.t. }\frac{t}{s}\rightarrow C.\label{17}%
\end{equation}
At the same time, there exists a function $t_{0}\left(  s\right)
\rightarrow+\infty$ as $s\rightarrow\infty,$ s.t. for any function $t\left(
s\right)  \geq t_{0}\left(  s\right)  $%
\begin{equation}
\left\langle \sigma_{\left(  s,\pm t\left(  s\right)  \right)  }\right\rangle
^{\pm\frac{\pi}{2}}\rightarrow\pm m^{\ast}\left(  \beta\right)  \text{ as
}s\rightarrow\infty.\label{18}%
\end{equation}
In fact, $t_{0}\left(  s\right)  $ can be any function
such that $\frac{t_{0}\left(  s\right)  }{s^{2}}\rightarrow\infty$ as
$s\rightarrow\infty$.

\end{theorem}

\begin{theorem}
\label{t3} Let the temperature $T=\beta^{-1}$ be low enough. For $\theta
\in\left[  -\frac{\pi}{2},\frac{\pi}{2}\right]  $ the state $\left\langle
\ast\right\rangle ^{\theta}$ can be obtained as a limit,%
\[
\left\langle \ast\right\rangle ^{\theta}=\lim_{N\rightarrow\infty}\left\langle
\ast\right\rangle _{N}^{\theta_{N}},
\]
where $\theta_{N}$ is any sequence of angles satisfying the condition:
$\lim_{N\rightarrow\infty}\theta_{N}=\theta.$
\end{theorem}

\begin{theorem}
Let the temperature $T=\beta^{-1}$ be low enough, and $\mu$ be any half-plane
Gibbs state with boundary condition $\sigma^{\pm}.$ Then there is a (unique)
probability measure $d_{\mu}$ on the segment $\left[  -\frac{\pi}{2},\frac
{\pi}{2}\right]  $ such that%
\[
\mu=\int\left\langle \ast\right\rangle ^{\theta}d_{\mu}\left(  \theta\right)
.
\]
This implies that the family $\left\{  \left\langle \ast\right\rangle
^{\theta},\theta\in\left[  -\frac{\pi}{2},\frac{\pi}{2}\right]  \right\}  $ of
states coincides with the family of all extremal Gibbs states of the
half-plane Ising model with boundary condition $\sigma^{\pm}.$
\end{theorem}

\section{Zero temperature case}

\label{zerotemp}

We start by considering the simplest case of zero temperature. Here, some of
the above theorems have to be modified slightly. Namely, the states
$\left\langle \ast\right\rangle ^{+\frac{\pi}{2}}$ and $\left\langle
\ast\right\rangle ^{-\frac{\pi}{2}}$ of Theorem \ref{t2} become trivial; they
are each concentrated on a single ground state configuration, $\sigma\equiv+1$
or $\sigma\equiv-1.$ Another trivial state is $\left\langle \ast\right\rangle
^{0}$ -- it is also concentrated on a single configuration, $\sigma^{\theta
=0}$ which is $+1$ (resp., $-1$) for $t>0$ (resp., $t<0$); all other states
$\left\langle \ast\right\rangle ^{\theta}$ are supported on infinitely many
ground state configurations.

The zero-temperature case is simpler because the relevant configurations have
only open contours (i.e., there are no loop contours), which do not interact,
except for a non-intersection condition --- see below.

Let $V_{N}\subset\mathbb{Z}^{2}$ be the box%
\[
V_{N}=\left\{  \left(  x,y\right)  :0\leq x\leq N,-N\leq y\leq N\right\}
\]
and let
\[
V_{N}^{\ast}=\left\{  \left(  s,t\right)  \in\mathbb{Z}^{2\ast}%
:0<s<N,-N<t<N\right\}
\]
with boundary $\partial V_{N}^{\ast}$. A spin boundary condition
$\sigma_{\partial V_{N}^{\ast}}$ on $\partial V_{N}^{\ast}$
is specified by a collection of an even number of distinct points
$z_{1},...,z_{2k}\in\partial V_{N},$ where the configuration
$\sigma_{\partial V_{N}^{\ast}}$ changes its value from $\pm1$ to $\mp1$. In
our case we can put $z_{1}$ to be the point $\left(  0,0\right)  ,$ and we
have to suppose that all the other points $z_{i}$ are in
\[
\partial^{-}V_{N}\,=\,\partial V_{N}\setminus\left\{  \left(  x,y\right)
:x=0,-N<y<N\right\}  .
\]
Every spin configuration $\sigma$ on $V_{N}^{\ast}$ with this boundary
condition defines a partition $p$ of the set $\left\{  z_{1},...,z_{2k}%
\right\}  $ into
pairs $\left\{  z_{i},z_{p\left(  i\right\}  }\right\}  $, in the following
way: among the Peierls contours of $\sigma$ there are precisely $k$ open
contours $\gamma_{i},$ $1\leq i\leq k,$ with $\cup_{i}\partial\left(
\gamma_{i}\right)  =\left\{  z_{1},...,z_{2k}\right\}  ;$ the rest of the
contours of $\sigma$ are closed contours, i.e., loops. Then the partition into
pairs is defined by
\[
\left\{  z_{i},z_{p\left(  i\right\}  }\right\}  =\partial\left(  \gamma
_{i}\right)  .
\]
If the configuration $\sigma$ is a ground state configuration, then it has no
loops, while the corresponding partition has minimal length: for any other
partition~$q$ we have%
\[
\sum_{i}\left\vert z_{i}-z_{p\left(  i\right)  }\right\vert _{1}\leq\sum
_{i}\left\vert z_{i}-z_{q\left(  i\right)  }\right\vert _{1},
\]
where $\left\vert \ast\right\vert _{1}$ denotes $L_{1}$-distance. Moreover,
\[
\sum_{i}\left\vert \gamma_{i}\right\vert =\sum_{i}\left\vert z_{i}-z_{p\left(
i\right)  }\right\vert _{1}.
\]
Without loss of generality we can suppose that the indices are chosen such
that the partition $\left\{  \left\{  z_{1},z_{2}\right\}  ,...,\left\{
z_{2k-1},z_{2k}\right\}  \right\}  $ is a ground state partition.

We will argue now that all the straight-line segments $\left[  z_{3}%
,z_{4}\right]  ,...,\left[  z_{2k-1},z_{2k}\right]  \subset\mathbb{R}^{2}$ are
far away from the segment $\left[  z_{1},z_{2}\right]  $ `in the bulk'. It is
enough to consider the case $k=2.$ Suppose the point $w\in\left[  z_{1}%
,z_{2}\right]  $ at ($L_{2}$) distance at least $c_{1}N$ from $z_{2},$ is at
distance $c_{2}N$ from the segment $\left[  z_{3},z_{4}\right]  ,$ see Fig. 1.
If $\arcsin\frac{c_{2}}{c_{1}}<\alpha,$ then $\alpha$ has to exceed
$\arctan\frac{1}{2}$ -- otherwise the ground state condition is violated.
Hence, $c_{2}>c_{1}\sin\left(  \arctan\frac{1}{3}\right)  .$

\begin{figure}[ptb]
\centering
\includegraphics[width=8cm]{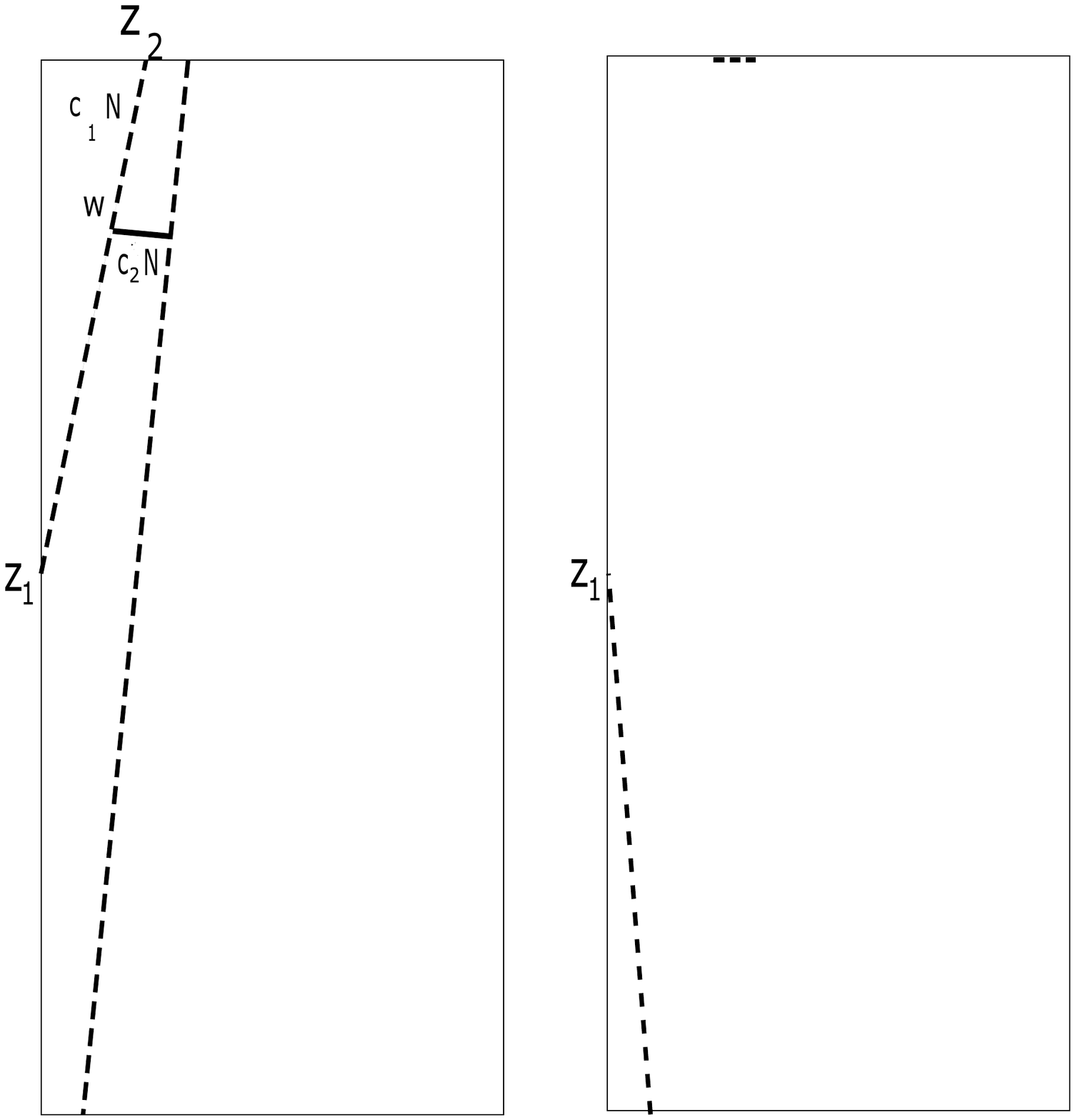}
\caption{The surgery on long paths.}%
\end{figure}

By the same token, all the segments $\left[  z_{2l-1},z_{2l}\right]  $ with
$l>1$ cannot pass too close to the origin:%
\[
\left[  z_{2l-1},z_{2l}\right]  \cap V_{N/4}=\varnothing.
\]

Let now $\left\langle \ast\right\rangle _{\left\{  z_{1},z_{2}\right\}
,...,\left\{  z_{2k-1},z_{2k}\right\}  }$ be the ground state in $V_{N},$
corresponding to the ground state partition $\left\{  \left\{  z_{1}%
,z_{2}\right\}  ,...,\left\{  z_{2k-1},z_{2k}\right\}  \right\}  ,$ and
$\sigma$ be a ground state spin configuration from that state $\left\langle
\ast\right\rangle _{\left\{  z_{1},z_{2}\right\}  ,...,\left\{  z_{2k-1}%
,z_{2k}\right\}  }$. Let $\gamma_{1},...,\gamma_{k}$ be its collection of open
contours. Then the probability in $\left\langle \ast\right\rangle _{\left\{
z_{1},z_{2}\right\}  ,...,\left\{  z_{2k-1},z_{2k}\right\}  }$
\[
\mathbf{\Pr}\left(  \left\{  \gamma_{2},...,\gamma_{k}\right\}  \cap
V_{N/4}\neq\varnothing\right)  \rightarrow0
\]
as $N\rightarrow\infty$, because for each $\varepsilon>0$ and each $l$ the
probability
\begin{equation}
\mathbf{\Pr}\left(  \mathrm{dist}\left(  \gamma_{l},\left(  z_{2l-1}%
,z_{2l}\right)  \right)  >N^{1/2+\varepsilon}\right)  \rightarrow0\text{ as
}N\rightarrow\infty. \label{11}%
\end{equation}

Summarizing, we can say that the distance between the projected states,
\[
\mathrm{dist}\left(  \left\langle \ast{\Huge |}_{V_{N/4}}\right\rangle
_{\left\{  z_{1},z_{2}\right\}  ,...,\left\{  z_{2k-1},z_{2k}\right\}
},\left\langle \ast{\Huge |}_{V_{N/4}}\right\rangle _{\left\{  z_{1}%
,z_{2}\right\}  }\right)  \rightarrow0
\]
as $N\rightarrow\infty,$ where\textbf{ }$z_{i}=z_{i}\left(  N\right)
,i=2,...,2k\in\partial^{-}V_{N}.$ Therefore, all the ground states of our
model are among the limit points of the states $\left\langle \ast\right\rangle
_{\left\{  z_{1},z_{2}\left(  N\right)  \right\}  }.$

In the case when the sequence $z_{2}\left(  N\right)  $ is in the ray
$r_{\theta}=\left\{  \left(  s,t\right)  :t=s\tan\theta\right\}  $ (or,
rather, to its neighborhood of radius $1/2,$ due to the rounding off to the
closest integer point) the existence of the limit ground state $\left\langle
\ast\right\rangle ^{\theta}=\lim_{N\rightarrow\infty}$ $\left\langle
\ast\right\rangle _{\left\{  z_{1},z_{2}\left(  N\right)  \right\}  }$ is
evident. It is supported by the set of infinite staircase contours $\gamma$
starting from $\left(  0,0\right)  $ and having asymptotic slope $\theta.$ The
probability in the state $\left\langle \ast\right\rangle ^{\theta}$ that the
initial piece $\gamma_{N}=\gamma\cap Z_{N}^{2}$ of the contour $\gamma$ has
$n$ vertical steps at $n$ prescribed locations (not necessarily different) on
the $Ox$ axis is given by $p^{n}\left(  1-p\right)  ^{N},$ with $p=p\left(
\theta\right)  =\tan\theta/\left(  1+\tan\theta\right)  .$

In case the sequence $\frac{z_{2}\left(  N\right)  }{N}$ has at least two
different subsequence limit points, the sequence of ground states
$\left\langle \ast\right\rangle _{\left(  z_{1},z_{2}\left(  N\right)
\right)  }$ also has at least two different limit points, due to $\left(
\ref{11}\right)  ,$and so does not have a unique limit.

Finally, suppose the sequence $\frac{z_{2}\left(  N\right)  }{N}$ has a limit,
and so defines the limiting ray $r_{\theta},$ $\theta=\theta\left[  \left(
z_{2}\left(  \ast\right)  \right)  \right]  .$ Let us show that $\lim
_{N\rightarrow\infty}$ $\left\langle \ast\right\rangle _{\left(  z_{1}%
,z_{2}\left(  N\right)  \right)  }=\left\langle \ast\right\rangle ^{\theta}.$
We suppose additionally that $\theta\in\left(  0,\frac{\pi}{4}\right)  ,$
since the case where $\theta\in\lbrack\pi/4,\pi/2)$ can be handled similarly.
Let $\bar{z}_{2}\left(  N\right)  \equiv\left(  \bar{a}\left(  N\right)
,\bar{b}\left(  N\right)  \right)  =\lfloor r_{\theta}\cap\partial
V_{N}\rfloor$ be the `integer part' of $r_{\theta}\cap\partial V_{N}$, and let
$z_{2}\left(  N\right)  \equiv\left(  a\left(  N\right)  ,b\left(  N\right)
\right)  =\bar{z}_{2}\left(  N\right)  +o\left(  N\right)  $ as $N\rightarrow
\infty.$ Let us fix a point $\left(  m,n\right)  \in\mathbb{R}^{2}$ with
$m,n>0;$ for all $N$ large enough we have $\left(  m,n\right)  \in V_{N}.$ To
simplify our exposition, we consider the case when $a\left(  N\right)
=\bar{a}\left(  N\right)  ,$ while $b\left(  N\right)  \geq\bar{b}\left(
N\right)  ,$ so $b\left(  N\right)  =\bar{b}\left(  N\right)  +o\left(
N\right)  .$ The relevant `ratio of partition functions' in our case is just
the ratio of binomial coefficients,%
\[
\frac{\dbinom{a\left(  N\right)  +b\left(  N\right)  -m-n}{b\left(  N\right)
-n}}{\dbinom{a\left(  N\right)  +\bar{b}\left(  N\right)  -m-n}{\bar{b}\left(
N\right)  -n}},
\]
which goes to $1$ as $N\rightarrow\infty,$ for any $m,n$ (though not uniformly
in $m,n$) in case $b\left(  N\right)  =\bar{b}\left(  N\right)  +o\left(
N\right)  .$ That follows from Stirling's formula.

In words, the reason for the identity $\lim_{N\rightarrow\infty}$
$\left\langle \ast\right\rangle _{\left(  z_{1},z_{2}\left(  N\right)
\right)  }=\left\langle \ast\right\rangle ^{\theta}$ is that the probability
of having the `extra' $o\left(  N\right)  $ (uniformly distributed) vertical
segments (which one has to add to $\bar{b}\left(  N\right)  $ in order to get
$b\left(  N\right)  $ of them) to be located at any of the first $m$ positions
of the segment $\left[  1,a\left(  N\right)  \right]  $, goes to zero as
$N\rightarrow\infty.$

\section{Low temperature case \label{lowtemp}}

For $T>0$ the problem becomes more involved, since the contours $\gamma
_{1},...,\gamma_{k}$ are interacting. We first describe their joint
distribution. The following formula follows from the cluster expansion
technique; it can be found in \cite{DKS}. The weight of our family
$\Gamma=\left\{  \gamma_{1},...,\gamma_{k}\right\}  $ is given by
\begin{equation}
w\left(  \Gamma\right)  =\exp\left\{  -\beta\left\vert \Gamma\right\vert
+\sum_{\substack{C\subset V_{N}:\\C\cap\Delta_{\Gamma}\neq\varnothing}%
}\Phi\left(  C\right)  \right\}  , \label{10}%
\end{equation}
where $\beta=\frac{1}{T}$ is the inverse temperature, $\left\vert
\Gamma\right\vert =\left\vert \gamma_{1}\right\vert +...+\left\vert \gamma
_{k}\right\vert $ is the total length of our contours, while the term
$\sum_{C:C\cap\Delta_{\Gamma}\neq\varnothing}\Phi\left(  C\right)  ,$ which we
explain now, contains the interaction between $\gamma$'s. The sum is taken
over all connected subsets ($\equiv$clusters) $C\subset V_{N}$. The notation
$C\cap\Delta_{\gamma}\neq\varnothing$ essentially means that $C$ intersects
the union $\Gamma$ of the contours $\gamma_{i}$, while the function
$\Phi\left(  C\right)  $ (which, of course, depends also on $\beta$) has the
following properties:

\begin{itemize}
\item Decay: for all $\beta$ sufficiently large,
\begin{equation}
\left\vert \Phi\left(  C\right)  \right\vert \leq\exp\left\{  -2{\beta
}(\mathrm{diam}_{\infty}(C)+1)\right\}  . \label{07}%
\end{equation}

\item Symmetry: the function $\Phi$ is translation invariant.
\end{itemize}

In what follows we will rely on the results of~ \cite{DKS}. However, there are
some important differences between our situation and the one treated there. In
\cite{DKS} the analysis is restricted to the case of periodic boundary
conditions, while here our contours $\gamma$ are in the box $V_{N}$ and so
interact with its boundary $\partial V_{N}$: -- inspecting the relation
$\left(  \ref{10}\right)  $, we see that the part of the contour in the bulk
lives in a different potential landscape than the part near the boundary.
Indeed, let $\lambda$ be some fragment of $\gamma,$ and suppose the cluster
$C\subset V_{N}$ intersects $\lambda$ and thus contributes to $\left(
\ref{10}\right)  $ an amount $\Phi\left(  C\right)  .$ Let us shift $\lambda$
by a vector $s\in\mathbb{Z}^{2},$ in such a way that $\lambda+s$ is still in
$V_{N},$ but $C+s$ is not. Then the corresponding contribution $\Phi\left(
C+s\right)  \left(  =\Phi\left(  C\right)  \right)  $ is missing near the
boundary. In case $\Phi\left(  C\right)  <0$ that would mean that $\gamma$ is
effectively attracted to the wall $\partial V_{N}$. The study of this issue
turns out to be quite complicated technically; it is done in \cite{IST}. A
main result of \cite{IST} is that even if such an attraction is present, it is
beaten by the entropic repulsion of $\gamma$ from the wall $\partial V_{N},$
provided $\left(  \ref{07}\right)  $ holds. (In fact, the result of \cite{IST}
is more precise: the entropic repulsion beats an attraction of strength
$-\exp\left\{  -\chi{\beta}(\mathrm{diam}_{\infty}(C)+1)\right\}  $ provided
$\chi>\frac{1}{2},$ but can fail against it for $\chi<\frac{1}{2}.$)

Let us show that the states $\left\langle \ast\right\rangle ^{\theta}$ are
distinct for different $\theta\in\left(  -\frac{\pi}{2},\frac{\pi}{2}\right)
.$ For that we will use the following result of \cite{DKS}, contained in
Theorem~4.16 there (see also relation~(2.5) and the Theorem 29 of \cite{IST},
where an error in \cite{DKS} is corrected). To lighten the notation, we
consider only the case when $\theta\in\left[  -\frac{\pi}{4},\frac{\pi}%
{4}\right]  .$

Let us introduce the following cigar-shaped subset $U_{N,d,\varkappa}\subset
V_{N}:$%
\[
U_{N,d,\varkappa}=\left\{  \left(  x,y\right)  \in V_{N}:\left\vert
y-x\tan\theta\right\vert \leq d\left(  \frac{x\left(  N-x\right)  }{N}\right)
^{\frac{1}{2}+\varkappa}\right\}  ,
\]
where $d,\varkappa>0.$ Let
\[
U_{N,d,\varkappa,R}=U_{N,d,\varkappa}\cup D_{R}\left(  0,0\right)  \cup
D_{R}\left(  N,N\tan\theta\right)
\]
be the union of $U_{N,d,\varkappa}$ and two disks of radius $R$ centered at
the endpoints of the line segment $((0,0),(N, N\tan\theta))$.

Every configuration from the state $\left\langle \ast\right\rangle
_{N}^{\theta}$ possesses exactly one open contour, which we denote by
$\gamma.$ Then for any $\varkappa>0$ the probability of the event
\[
\gamma\subset U_{N,d,\varkappa,R},
\]
computed in this state $\left\langle \ast\right\rangle _{N}^{\theta},$ goes to
$1,$ as $R$ increases, uniformly in $N.$ That shows that $\left\langle
\ast\right\rangle ^{\theta_{1}}\neq\left\langle \ast\right\rangle ^{\theta
_{2}}$ for different $\theta$'s.

Now we will show that $\lim_{N\rightarrow\infty}\left\langle \ast\right\rangle
_{N}^{\theta^{\prime}\left(  N\right)  }=\left\langle \ast\right\rangle
^{\theta}$ if $\lim_{N\rightarrow\infty}\theta^{\prime}\left(  N\right)
=\theta.$ To see that we will check that the states $\left\langle
\ast\right\rangle ^{\theta^{\prime}}$ and $\left\langle \ast\right\rangle
^{\theta}$ are absolutely continuous with respect to each other. So we
introduce the partition function%
\begin{align*}
Z\left(  \theta,N\right)   &  =\sum_{\gamma:\left(  0,0\right)
\rightsquigarrow\left(  N,N\tan\theta\right)  }w\left(  \gamma\right)  \\
&  \equiv\sum_{\gamma:\left(  0,0\right)  \rightsquigarrow\left(
N,N\tan\theta\right)  }\exp\left\{  -\beta\left\vert \gamma\right\vert
+\sum_{\substack{C\subset V_{N}:\\C\cap\Delta_{\gamma}\neq\varnothing}%
}\Phi\left(  C\right)  \right\}  ,
\end{align*}
where the sum is taken over all paths $\gamma$ in $V_{N},$ connecting $\left(
0,0\right)  $ to $\left(  N,N\tan\theta\right)  .$ We denote by $Z\left(
\theta^{\prime},N\right)  $ the partition function $Z\left(  \theta^{\prime
}\left(  N\right)  ,N\right)  .$ Let us fix an integer $n,$ and let $\eta$ be
a path in $V_{n},$ connecting $\left(  0,0\right)  $ to a point $\left(
k,l\right)  \in\partial V_{n}.$ We will need the partition functions%
\[
Z\left(  \theta,N,\eta\right)  =\sum_{\substack{\gamma:\left(  k,l\right)
\rightsquigarrow\left(  N,N\tan\theta\right)  \\\eta\cap\gamma=\varnothing
}}w\left(  \eta\cup\gamma\right)  ,
\]
where $\eta\cap\gamma=\varnothing$ means that the contours $\gamma$ and $\eta$
are compatible, so that their concatenation $\eta\cup\gamma$ is a legitimate
contour, $\eta\cup\gamma:\left(  0,0\right)  \rightsquigarrow\left(
N,N\tan\theta\right)  .$ The partition function $Z\left(  \theta^{\prime
},N,\eta\right)  $ is defined in the obvious way. Our claim boils down to
showing that the cross-ratio%
\begin{equation}
\frac{Z\left(  \theta^{\prime},N,\eta\right)  }{Z\left(  \theta^{\prime
},N\right)  }\frac{Z\left(  \theta,N\right)  }{Z\left(  \theta,N,\eta\right)
}\label{13}%
\end{equation}
is bounded from above and below, uniformly in $N.$

The bounds on the above partition functions, provided by the analysis of
\cite{DKS} -- see relations (4.11.3), (4.12.3) there -- show that
\begin{equation}
\ln Z\left(  \theta,N\right)  =-\beta\tau_{\beta}\left(  \theta\right)
\frac{N}{\cos\theta}-\frac{1}{2}\ln\frac{N}{\cos\theta}+O\left(  1\right)
,\label{14}%
\end{equation}
where $\tau_{\beta}\left(  \ast\right)  $ is the surface tension function.
(Unlike the situation of a 1D Gibbs field with finite spin state space, where
the log of the partition function $\ln Z\left(  N\right)  =-\beta f+O\left(
1\right)  ,$ here we have the additional \textit{universal }
`Ornstein-Zernike' term $-\frac{1}{2}\ln N,$ see \cite{CIV}.) Plugging
$\left(  \ref{14}\right)  $ into $\left(  \ref{13}\right)  $ and using
analyticity of the function $\tau_{\beta}\left(  \ast\right)  $ in $\theta,$
we see that the terms which grow in $N$ cancel each other; thus we establish
the boundedness of $\left(  \ref{13}\right)  $ in $N.$ (Of course, it is not
uniform in $\eta.$)

The same argument shows that the cross-ratio%
\[
\frac{Z\left(  \theta,N_{1},\eta\right)  }{Z\left(  \theta,N_{1}\right)
}\frac{Z\left(  \theta,N_{2}\right)  }{Z\left(  \theta,N_{2},\eta\right)  }%
\]
is bounded from above and below uniformly in $N_{1},N_{2}$ (again, not
uniformly in $\eta$), which shows the existence of the weak limit
$\lim_{N\rightarrow\infty}\left\langle \ast\right\rangle _{N}^{\theta
}=\left\langle \ast\right\rangle ^{\theta}.$

Finally we show that every state of our system is a mixture of various
$\left\langle \ast\right\rangle ^{\theta}$'s. This will be obtained as an
adaptation of the arguments from the previous section. Namely, we will show
that the contour $\gamma_{1}$ is the only contour visible in any finite
vicinity of the point $\left(  0,0\right)  ,$ with overwhelming probability as
$N\rightarrow\infty.$ To that end we will show that the total length of the
collection $\gamma_{1},...,\gamma_{k}$ is sufficiently close to its minimal
possible value, with high probability, provided the temperature is low enough.
That would imply that among the contours $\gamma_{1},...,\gamma_{k}$ there is
only one -- $\gamma_{1}$ -- which comes into the vicinity of the point
$\left(  0,0\right)  .$ To this end we introduce another collection of open
Ising contours, $\nu_{1},...,\nu_{k}$ in $V_{N},$ which has the same set of
end-points $z_{1},...,z_{2k},$ and which has minimal total length $\left\vert
\nu_{1}\right\vert +...+\left\vert \nu_{k}\right\vert $ among all such
collections. (In other words, the collection $\nu_{1},...,\nu_{k}$ defines the
ground state spin configuration.)

\begin{lemma}
Let $s>0,$ and $L>\left(  \left\vert \nu_{1}\right\vert +...+\left\vert
\nu_{k}\right\vert \right)  \left(  1+s\right)  .$ Then there exists
$\beta\left(  s\right)  $ such that%
\[
\mathbf{\Pr}\left(  \left\vert \gamma_{1}\right\vert +...+\left\vert
\gamma_{k}\right\vert >L\right)  \leq\exp\left\{  -\left(  \beta-\beta\left(
s\right)  \right)  \left[  L-\left(  \left\vert \nu_{1}\right\vert
+...+\left\vert \nu_{k}\right\vert \right)  \right]  \right\}  .
\]

\end{lemma}

\begin{proof}
Consider the set of bonds $D$, which is the symmetric difference $\left\{
\gamma_{1}\cup...\cup\gamma_{k}\right\}  \bigtriangleup\left\{  \nu_{1}%
\cup...\cup\nu_{k}\right\}  .$ $D$ is a collection of closed contours, and if
$\left\vert \gamma_{1}\right\vert +...+\left\vert \gamma_{k}\right\vert >L,$
then $\left\vert D\right\vert >L-\left(  \left\vert \nu_{1}\right\vert
+...+\left\vert \nu_{k}\right\vert \right)  .$ Now we define the Peierls
transformation, which from every configuration $\sigma$ containing the
contours $\gamma_{1},...,\gamma_{k}$ produces another configuration $\pi
_{D}\left(  \sigma\right)  ,$ which satisfies the same boundary condition
$\left\{  z_{1},...,z_{2k}\right\}  $ but whose energy is smaller by at least
$L-\left(  \left\vert \nu_{1}\right\vert +...+\left\vert \nu_{k}\right\vert
\right)  .$ Moreover, if $\sigma\neq\sigma^{\prime},$ then $\pi_{D}\left(
\sigma\right)  \neq\pi_{D}\left(  \sigma^{\prime}\right)  .$ The construction
is the following: if $\sigma$ corresponds to the set $B\left(  \sigma\right)
$ of bonds forming all the contours of $\sigma,$ then $\pi_{D}\left(
\sigma\right)  $ corresponds to the bond set $B\left(  \sigma\right)
\bigtriangleup D.$ That correspondence proves the estimate%
\[
\mathbf{\Pr}\left(  \left\{  \sigma:\gamma_{1}\cup...\cup\gamma_{k}\subset
B\left(  \sigma\right)  \right\}  \right)  \leq\exp\left\{  -\beta\left(
\left\vert \gamma_{1}\right\vert +...+\left\vert \gamma_{k}\right\vert
-\left\vert \nu_{1}\right\vert -...-\left\vert \nu_{k}\right\vert \right)
\right\}  .
\]

Assuming $L>\left(  \left\vert \nu_{1}\right\vert +...+\left\vert \nu
_{k}\right\vert \right)  \left(  1+s\right)  $ we have $L-\left(  \left\vert
\nu_{1}\right\vert +...+\left\vert \nu_{k}\right\vert \right)  >s^{\prime}L$
with $s^{\prime}>0.$ Therefore the entropy factor $3^{L}$ is beaten by the
energy gain $L-\left(  \left\vert \nu_{1}\right\vert +...+\left\vert \nu
_{k}\right\vert \right)  ,$ and the proof follows.
\end{proof}

From the above lemma we conclude that the contours $\gamma_{2},...,\gamma_{k}$
stay close to the ground state contours, so the results of the previous
section apply.

The only items left to prove are the properties of the states $\left\langle
\ast\right\rangle ^{\pm\frac{\pi}{2}}.$ The relation $\left(  \ref{17}\right)
$ is proven in the same way as was used in proving that the states
$\left\langle \ast\right\rangle ^{\theta_{1}}$ and $\left\langle
\ast\right\rangle ^{\theta_{2}}$ are different. What is needed for $\left(
\ref{18}\right)  $ is that the contour $\gamma_{1},$ distributed according to
the state $\left\langle \ast\right\rangle ^{+\frac{\pi}{2}}$ fluctuates away
from $Oy$ (the positive $y$-axis). The precise way in which it does this can
be analysed by using the \textit{methods }of the paper \cite{IST}. But even
the \textit{results }of~\cite{IST} show that it does fluctuate away from $Oy.$
In what follows we will prove one version of this phenomenon. First, we
introduce some notation.

Let $\gamma\in V_{N}$ be a contour connecting the points $\left(  0,0\right)
$ and $\left(  0,N\right)  .$ For every integer $h$ define the set
$Y_{N,h}\left(  \gamma\right)  $ by%
\[
Y_{N,h}\left(  \gamma\right)  =\left\{  y\in\left[  0,N\right]  :\max\left\{
m:\left[  0,m\right]  \cap\gamma=\varnothing\right\}  >h\right\}  .
\]
In words, $Y_{N,h}\left(  \gamma\right)  $ is the set of locations where the
contour $\gamma$ is farther than $h$ from $Oy.$

\begin{proposition}
For all temperatures low enough and for every $h>0$ we have
\begin{equation}
\mathbb{E}_{N}\left(  \frac{\left\vert Y_{N,h}\left(  \gamma_{1}\right)
\right\vert }{N}\right)  \rightarrow1\text{ as }N\rightarrow\infty,\label{21}%
\end{equation}
where the expectation is computed in the state $\left\langle \ast\right\rangle
_{N}^{+\frac{\pi}{2}}.$
\end{proposition}

\begin{proof}
As we know from the main result of \cite{IST}, the limit%
\begin{align}
&  \lim_{N\rightarrow\infty}-\frac{\ln\left(  \sum\limits_{\substack{\gamma
_{1}\subset V_{N}:\\\gamma_{1}:\left(  0,0\right)  \rightsquigarrow\left(
0,N\right)  }}\exp\left\{  -\beta\left\vert \gamma_{1}\right\vert
+\sum\limits_{\substack{C\subset V_{N}:\\C\cap\Delta_{\gamma_{1}}%
\neq\varnothing}}\Phi\left(  C\right)  \right\}  \right)  }{\beta N}%
\label{20}\\
&  =\tau_{\beta}\left(  \frac{\pi}{2}\right)  \, ,\nonumber
\end{align}
where $\tau_{\beta}\left(  \frac{\pi}{2}\right)  $ is the surface tension of
the Ising model. This is so despite the fact that the contour $\gamma_{1}$ is
confined to the right halfplane and despite the expression involving the
clusters $C$ entering $\left(  \ref{20}\right)  .$ Moreover, the relation
$\left(  \ref{20}\right)  $ would stay true even if we supplement the set of
allowed clusters by some extra clusters -- by $\bar{C}$'s, intersecting
$\Delta_{\gamma_{1}}$ but not confined to the right halfplane. These auxiliary
clusters $\bar{C}$ can have positive weights $\Phi\left(  \bar{C}\right)  ,$
thus introducing an extra attraction to the $Oy$ axis. Still, the relation
$\left(  \ref{20}\right)  $ holds, provided that%
\[
\left\vert \Phi\left(  \bar{C}\right)  \right\vert <\exp\left\{  -\chi{\beta
}(\mathrm{diam}_{\infty}(\bar{C})+1)\right\}  \text{ with }\chi>\frac{1}{2}.
\]
In particular, let us add to the set of clusters $\left\{  C\subset
V_{N}:C\cap\Delta_{\gamma_{1}}\neq\varnothing\right\}  $ extra clusters
$\bar{C},$ which are simply horizontal segments of length $h,$ which intersect
both the line $Oy$ and the contour $\gamma_{1}.$ Define their weight to be
$\exp\left\{  -{\beta}(h+1)\right\}  .$ According to \cite{IST}, we still have
the same limit:%

\begin{align*}
&  \lim_{N\rightarrow\infty}-\frac{1}{\beta N}\ln\left(  \sum
\limits_{\substack{\gamma_{1}\subset V_{N}:\\\gamma_{1}:\left(  0,0\right)
\rightsquigarrow\left(  0,N\right)  }}\right. \\
&  \left.  \exp\left\{  -\beta\left\vert \gamma_{1}\right\vert +\sum
\limits_{\substack{C\subset V_{N}:\\C\cap\Delta_{\gamma_{1}}\neq\varnothing
}}\Phi\left(  C\right)  +\exp\left\{  -{\beta}(h+1)\right\}  \left(
N-\left\vert Y_{N,h}\left(  \gamma_{1}\right)  \right\vert \right)  \right\}
\right) \\
&  =\tau_{\beta}\left(  \frac{\pi}{2}\right)  ,
\end{align*}
which proves $\left(  \ref{21}\right)  .$
\end{proof}

\section{Main theorems valid for all $T < T_{c}$}

\label{theorems2}

We define for each $T=\beta^{-1}>0$ and $\theta\in\lbrack-\pi/2,+\pi/2]$ a
collection $G_{\theta}^{T}$ of Gibbs states of the model~(\ref{01}) at
temperature~$T$ as follows. For $\theta\in(-\pi/2,+\pi/2)$ (resp., for
$\theta=\pm\pi/2$), $\left\langle \ast\right\rangle $ is in $G_{\theta}^{T}$
if (\ref{1.5}) is valid (resp., (\ref{17}) is valid for any $C\in
(-\infty,+\infty)$).

\begin{theorem}
\label{theorem6-1} Let $T \in(0,T_{c})$; then $G_{\theta}^{T}$ is nonempty for
every $\theta\in[-\pi/2, +\pi/2]$. States from $G_{\theta}^{T}$ and
$G_{\theta^{\prime}} ^{T}$ with $\theta\neq\theta^{\prime}$ are mutually singular.
\end{theorem}

The next theorem is analogous to Theorem~\ref{t3} for low $T$ given above,
except that now (a) we do not know that all subsequence limits agree and (b)
we treat $\theta= \pm\pi/2$ differently. The proofs of both theorems are then
given together.

\begin{theorem}
\label{theorem6-2} Let $T\in(0,T_{c})$ and let $\theta_{N}$ and $\left\langle
\ast\right\rangle _{N}^{\theta_{N}}$ be as in Theorem~\ref{t3}. Then for
$\theta\in(-\pi/2,+\pi/2)$, every subsequence limit of of $\left\langle
\ast\right\rangle _{N}^{\theta_{N}}$ belongs to $G_{\theta}^{T}$. Let
$\theta_{m}$ be a sequence in $(-\pi/2,+\pi/2)$ converging to $\pm\pi/2$ and
let $\left\langle \ast\right\rangle ^{\theta_{m}}\in G_{\theta_{m}}^{T}$. Then
any subsequence limit as $m\rightarrow\infty$ of $\left\langle \ast
\right\rangle ^{\theta_{m}}$ is in $G_{\pm\pi/2}^{T}$.
\end{theorem}

\begin{proof}
The proof of Theorems~\ref{theorem6-1} and~\ref{theorem6-2}
is based on the planar Ising model exact calculations for the profile of an
interface at angle $\theta$ given in the next section of the paper; earlier
exact calculations on interface profiles may be found in~\cite{AR,FFW,AU}. For
those calculations, we note that one may obtain the states $\left\langle
\ast\right\rangle _{N}^{\theta_{N}}$ in the strip that is infinite in the
vertical $t$-coordinate (and of width $N$ in the horizontal $s$-coordinate) by
the following procedure. First take the finite region that is periodic in the
$t$-coordinate with large but finite period~$K$ and choose boundary conditions
on the left and right boundaries (at $s=0$ and $s=N$) which require two
interfaces
as follows: all spins are $+1$ at both the left and right boundaries
\textit{except\/} for $-1$ spins when $t$ is between a large negative value
$M$ and~$0$ (resp., between $M$ and $N\tan{\theta}$) on the left (resp., the
right) boundary. Then take the limit where first $K\rightarrow\infty$ and
then
$M\rightarrow-\infty$ so that the
$M$-interface is eliminated, resulting in the $\left\langle \ast\right\rangle
_{N}^{\theta_{N}}$ state on the
$N\times\infty$ strip.

Now the results of the next section, in particular Equations (\ref{D22}%
)--(\ref{D26}) give an exact formula for the extension of the limit in
(\ref{1.5}) when $t-s \tan{\theta}$ is proportional to $\sqrt{s}$. When $(t-s
\tan{\theta})/\sqrt{s} \to\pm\infty$, it corresponds to $z = \pm\infty$ in
(\ref{D24})--(\ref{D25}). This proves the first parts of
Theorems~\ref{theorem6-1} and~\ref{theorem6-2} for $\theta\in(-\pi/2, +\pi
/2)$. The rest of the claims of the theorems then follow from the well-known
result about the full plane Ising model for $T<T_{c}$ that the only pure
states are the plus and minus ones~\cite{A,H}.
\end{proof}

\section{Exact results for all $T < T_{c}$}

\label{belowcritical}


It turns out that Dobrushin boundary conditions in the half-plane Ising
ferromagnet has a useful realization in the spinor language of Kaufman
\cite{K} and of Schultz, Mattis and Lieb \cite{SML}. To have translational
symmetry at one's disposal, which is crucial, it is necessary to wrap the
lattice on a cylinder; thus with plus boundary conditions, say, it is
necessary to induce two domain walls crossing the cylinder. We take the
diameter of the cylinder to infinity and then follow that by imposing infinite
separation between the two domain walls. We then focus on one of these domain
walls. We select a Dobrushin boundary condition, one that forces the interface
to cross at a given mean angle; then we investigate the behavior of the
expectation of the magnetization near one face of the cylinder.
\begin{figure}[ptbh]
\centering
\includegraphics[width=10cm]{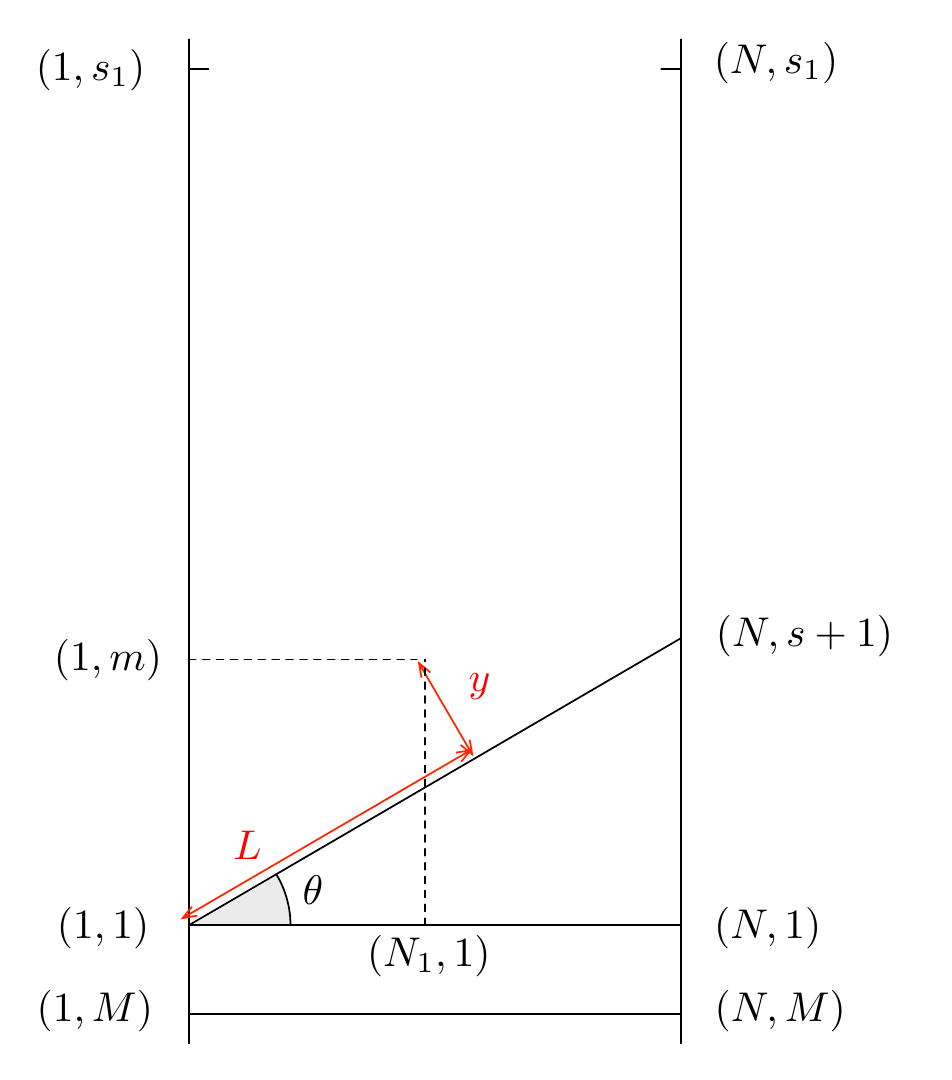}\caption{The diagonal interface and the
definitions used. Note that $m=[L\sin\theta+y\cos\theta]$, $N_{1}=[L\cos
\theta-y\sin\theta]$, $[\cdots]$ denotes the integer part.}%
\label{fig1}%
\end{figure}

The arrangement is shown in Fig. \ref{fig1}, the caption of which explains the
spatial coordinates. The partition function for such a domain wall, beginning
at $y=1$ and ending at, $y=s+1$ is given by
\begin{equation}
Z^{\times}=\frac{1}{2\pi}\int_{-\pi}^{\pi}\text{d}\omega\,\exp\Bigl[-N\gamma
(\omega)+is\omega\Bigr]|g(\omega)|^{2}\,.\label{D01}%
\end{equation}
Here, the function $\gamma(\omega)$ is one of the hyperbolic triangle elements
of Onsager~\cite{O}, given by
\begin{equation}
\cosh\gamma(\omega)=\cosh2K_{1}^{\star}\cosh2K_{2}-\sinh2K_{1}^{\star}%
\sinh2K_{2}\cos\omega\,,\,\,\gamma(\omega)\geqslant0\,,\omega\in
\mathbb{R}\,.\label{D02}%
\end{equation}
The interactions are in units of $k_{B}T$ and are taken different in the two
directions, something mathematically useful but physically irrelevant. The
variables $K_{j}^{\star}$, $j=1,2$ are dual ones given by
\begin{equation}
\exp\left(  2K_{j}^{\star}\right)  =\coth K_{j}\,.\label{D03}%
\end{equation}
The function $g$ has the same domain of analyticity as $\gamma$, a matter of
some significance, which we merely quote for the next step: we are interested
in an interface crossing at a given angle, say $\theta$, so we take
$N\rightarrow\infty$, $s=[N\tan\theta]$. The analytical tool here is saddle
point integration. The saddle point in $-\pi\leqslant\text{Re}(\omega
)\leqslant\pi$ is at $\omega_{s}=i\nu(\theta)$ where
\begin{equation}
\gamma^{(1)}\left(  i\nu(\theta)\right)  =i\tan\theta\,.\label{D04}%
\end{equation}
It is straightforward to see that there is a unique solution for
$0\leqslant\theta<\pi/2$ in the interval $0\leqslant\nu(\theta)<2|K_{1}%
-K_{2}^{\star}|$. The asymptotics of (\ref{D01}) are given by
\begin{equation}
Z^{\times}\sim|g\left(  i\nu(\theta)\right)  |^{2}\exp\Bigl[-L\tau
(\theta)\Bigr]\frac{1}{2\pi}\int_{-\infty}^{+\infty}\text{d}u\,\exp
\Bigl[-u^{2}L\cos\theta\gamma^{(2)}\left(  i\nu(\theta)\right)
/2\Bigr]\,.\label{D05}%
\end{equation}
The length $L$ is given by $N=[L\cos\theta]$ and the surface tension is
\begin{equation}
\tau(\theta)=\cos\theta\,\gamma\left(  i\nu(\theta)\right)  +\sin\theta
\,\nu(\theta)\,.\label{D06}%
\end{equation}
This is an exact derivation of the angle-dependent surface tension for
Dobrushin boundary conditions. The finite-size corrections are easily obtained
by integration in (\ref{D05}):
\begin{equation}
Z^{\times}\sim|g\left(  i\nu(\theta)\right)  |^{2}\exp\Bigl[-L\tau
(\theta)\Bigr]\Bigl[2\pi L\cos\theta\,\gamma^{(2)}\left(  i\nu(\theta)\right)
\Bigr]^{-1/2}\,.\label{D07}%
\end{equation}
We have just given the leading term, as soon we shall be interested only in
limiting behavior: (\ref{D05}) is equivalent to the first term of a Laplace
method. Some comments are in order: first, we have to get onto the steepest
descent path which actually goes monotone-upwards in the $\omega$ plane,
intersecting the line $\nu=\infty$, $\omega=u+iv$ at $-\pi<u<0$ to the left
and symmetrically at $0<u<\pi$ to the right. We have evaluated the function
$g$ and assure the reader that we do not have to cross a singularity of it to
get onto the steepest descent path. A good contemporary source on steepest
descent methods is Ablowitz and Fokas \cite{AF}. Another point, which must be
outlined here, is that the Fermionic realisation of the Dobrushin boundary
allows emission of more than one Fermion. These terms can be controlled; they
do not report in the limiting behavior given in (\ref{D01}), nor do they in
the magnetization profile, which we are about to specify. Consider the
magnetization at a position $(N_{1},m)$ in the Dobrushin boundary condition
used above. It is known that the leading term, when $M\rightarrow\infty$ and
when the second interface has been clustered away, is
\begin{equation}
\langle\sigma(N_{1},m)\rangle=\mathfrak{m}^{\ast}\left(  Z^{\times}\right)
^{-1}Y(N,N_{1},m,s)\,.\label{D08}%
\end{equation}
Here, $\mathfrak{m}^{\ast}$ is the spontaneous magnetization and $Y$ has the
integral representation:
\begin{align}
Y &  =\frac{1}{2\pi}\int_{-\pi}^{\pi}\text{d}\omega_{1}\,g(\omega
_{1})\,\text{e}^{-N_{1}\gamma(\omega_{1})}\label{D09}\\
\times &  \frac{\mathcal{P}}{2\pi}\int_{-\pi}^{\pi}\text{d}\omega
_{2}\,\text{e}^{-(N-N_{1})\gamma(\omega_{2})}\,\frac{j(\omega_{1},\omega_{2}%
)}{\text{e}^{i(\omega_{1}+\omega_{2})}-1}\text{e}^{-im(\omega_{1}+\omega
_{2})+is\omega_{2}}g^{\ast}(-\omega_{2})\,,\nonumber
\end{align}
in which $\ast$ denotes complex conjugation. This is what is left after
multiple domain wall configurations have been eliminated. The function
$j(\omega_{1},\omega_{2})$ is $2\pi$-periodic in both variables and has in
each variable the same domain of analyticity as $\sinh\gamma$, that is, square
root branch cuts at $\omega=\pm2i(K_{1}\pm K_{2}^{\star})$. Further, we have
$j(\omega,-\omega)=2$. We are interested in the limit as $N\rightarrow\infty$,
$s=[N\tan\theta]$, $0\leqslant\theta<\pi/2$; here, $[\cdots]$ denotes nearest
integer. We apply the Plemelj theorem, bearing in mind that we want to get
onto the steepest descent path for the $\omega_{2}$ integral; this is in the
upper half plane. We find that
\begin{equation}
\left(  Z^{\times}\right)  ^{-1}Y(N,N_{1},m,s)\sim1+\left(  Z^{\times}\right)
^{-1}X\,,\label{D10}%
\end{equation}
where
\begin{equation}
X=\frac{1}{2\pi}\int_{-\pi}^{\pi}\text{d}\omega_{1}\,g(\omega_{1}%
)\,\text{e}^{-N_{1}\gamma(\omega_{1})-im\omega_{1}}W(N,s,m|\omega
_{1})\,,\label{D11}%
\end{equation}
finally, we have:
\begin{equation}
W=\frac{1}{2\pi}\int_{\mathcal{C}}\text{d}\omega_{2}\,\text{e}^{-(N-N_{1}%
)\gamma(\omega_{2})}\text{e}^{i(s-m)\omega_{2}}\frac{j(\omega_{1},\omega_{2}%
)}{\text{e}^{i(\omega_{1}+\omega_{2})}-1}g^{\ast}(-\omega_{2})\,,\label{D12}%
\end{equation}
where $\mathcal{C}$ is exactly the steepest descent path used in the partition
function investigation above. We now evaluate the limit:
\begin{equation}
\lim\left(  Z^{\times}\right)  ^{-1}W=g^{\ast}\left(  -i\nu(\theta)\right)
\text{e}^{N_{1}\gamma\left(  i\nu(\theta)\right)  }\text{e}^{m\nu(\theta
)}\frac{j(\omega_{1},i\nu(\theta))}{\text{e}^{i\omega_{1}}\text{e}%
^{-\nu(\theta)}-1}\,,\label{D13}%
\end{equation}
thus we have the limiting result that $\left(  Z^{\times}\right)  ^{-1}$
converges to
\begin{equation}
\frac{1}{2\pi}\int_{-\pi}^{\pi}\text{d}\omega_{1}\,\text{e}^{-N_{1}%
\Bigl[\gamma(\omega_{1})-\gamma\left(  i\nu(\theta)\right)  \Bigr]}%
\text{e}^{-i\left(  \omega_{1}+i\nu(\theta)\right)  }\frac{j(\omega_{1}%
,i\nu(\theta))}{\text{e}^{i\omega_{1}}\text{e}^{-\nu(\theta)}-1}\frac
{g(\omega_{1})}{g\left(  i\nu(\theta)\right)  }\,.\label{D14}%
\end{equation}
We now investigate this by saddle-point integration: care is needed. The
integrand in (\ref{D14}) has a simple pole at $\omega_{1}=-i\nu(\theta)$. The
saddle point is given by
\begin{equation}
N_{1}\gamma^{(1)}\left(  i\nu_{s}(1)\right)  =-im\,.\label{D15}%
\end{equation}
It is natural to express the result in terms of the Euclidean distance along
the \textquotedblleft flattened\textquotedblright\ interface pointing in the
direction given by $\theta$ and the normal coordinate $y$. Thus we have
\begin{align}
N_{1} &  =L\cos\theta-y\sin\theta\,,\nonumber\label{D16}\\
m &  =L\sin\theta+y\cos\theta\,,
\end{align}
equation (\ref{D15}) becomes
\begin{equation}
\gamma^{(1)}\left(  i\nu_{s}(1)\right)  =-i\frac{L\sin\theta+y\cos\theta
}{L\cos\theta-y\sin\theta}\,.\label{D17}%
\end{equation}
Referring back to (\ref{D04}), we see that
\begin{equation}
\gamma^{(1)}\left(  i\nu_{s}(1)\right)  -\gamma^{(1)}\left(  i\nu
(\theta)\right)  =-i\frac{y}{L\cos^{2}\theta}\left(  1-\frac{y}{L}\tan
\theta\right)  ^{-1}\,.\label{D18}%
\end{equation}
We are interested in the case $L$ large and $\alpha\neq0$ where
\begin{equation}
y=\alpha L^{1/2}\,,\label{D19}%
\end{equation}
this means we can use a Taylor-series approximation on the left hand side of
(\ref{D18})
\begin{equation}
\nu_{s}(1)=-\nu(\theta)-\frac{\alpha}{L^{1/2}}\frac{1}{\cos^{2}\theta
\,\gamma^{(2)}\left(  i\nu(\theta)\right)  }+\mathcal{O}\left(  \frac{1}%
{L}\right)  \,.\label{D20}%
\end{equation}
Working up (\ref{D14}) with (\ref{D20}) must take careful account of the
simple pole and whether it is crossed in getting onto the steepest descents
contour. We encounter the following integral representation:
\begin{equation}
F(z)=\frac{1}{2\pi}\int_{-\infty}^{+\infty}\text{d}u\,\frac{\text{e}^{-u^{2}}%
}{z+iu}\,,\label{D21}%
\end{equation}
it is easy to see that $F(-z)=-F(z)$ and also to derive the identity
\begin{equation}
F(z)=\frac{2}{\sqrt{\pi}}\int_{z}^{+\infty}\text{d}u\,\text{e}^{-u^{2}%
}\,,\qquad z>0\,.\label{D22}%
\end{equation}
This applies to the development of (\ref{D14}) by steepest descents with
\begin{equation}
z=\frac{\alpha(\sec\theta)^{3/2}}{2\gamma^{(2)}\left(  i\nu(\theta)\right)
}\,.\label{D23}%
\end{equation}
With the definition
\begin{equation}
G(z)=\frac{2}{\sqrt{\pi}}\int_{0}^{z}\text{d}u\,\text{e}^{-u^{2}}\,,\qquad
z>0\,,\label{D24}%
\end{equation}
we have
\begin{equation}
s\lim\langle\sigma(N_{1},m)\rangle=-\mathfrak{m}^{\ast}\text{sgn}%
(z)G(|z|)\,,\label{D25}%
\end{equation}
where the rather complex limiting procedure $s\lim$ is specified in the text.
As a final remark, note that a full development of the saddle point equations
and the surface tension function (\ref{D06}) allow (\ref{D23}) to be expressed
in terms of thermodynamic quantities
\begin{equation}
z=\alpha\Bigl[\sec\theta\left(  \tau(\theta)+\tau^{(2)}(\theta)\right)
\Bigr]^{1/2}\,.\label{D26}%
\end{equation}
This combination of surface tension and derivatives is known as the surface
stiffness; it is interesting that such a meso-scale quantity, representing a
contraction of description, occurs here.


\end{document}